\documentclass[conference, 10pt]{IEEEtran} 

\usepackage[utf8]{inputenc} 
\usepackage[T1]{fontenc}    
\usepackage[colorlinks=true, linkcolor=red, citecolor=blue]{hyperref}       
\usepackage{booktabs}       
\usepackage{graphicx}
       
\usepackage{nicefrac}       
\usepackage{microtype}      
\usepackage[cmex10]{amsmath}
\usepackage{amssymb}
\usepackage{amsfonts}
\usepackage{thmtools}
\usepackage{amsthm}
\usepackage{cite}
\usepackage{subcaption}
\usepackage{tikz}
\usepackage[linesnumbered,ruled]{algorithm2e}
\newtheorem{definition}{Definition}
\newtheorem{theorem}{Theorem}
\newtheorem{lemma}{Lemma}
\newtheorem{example}{Example}

\makeatletter
\newcommand{\removelatexerror}{\let\@latex@error\@gobble}
\makeatother

\usepackage{color}

\title{Neural Mutual Information Estimation for Channel Coding: State-of-the-Art Estimators, Analysis, and Performance Comparison }

\author{
	  \authorblockN{Rick Fritschek$^{\ast}$, Rafael F. Schaefer$^{\dagger}$, and Gerhard Wunder$^{\ast}$\\[2mm]}
	  \IEEEauthorblockA{\begin{tabular}{cc}
	       \begin{tabular}{c}
	           $^{\ast}$ Heisenberg Communications and Information Theory Group\\
                        Freie Universit\"at Berlin, \\
                        Takustr. 9,
                        14195 Berlin, Germany\\
                        Email: \texttt{\{rick.fritschek, g.wunder\}@fu-berlin.de}
	       \end{tabular}
	       \begin{tabular}{c}
	           $^{\dagger}$ Information Theory and Applications Chair\\
	                        Technische Universit{\"a}t Berlin \\
                            Einsteinufer 25, 10587 Berlin, Germany\\
                            Email: \texttt{rafael.schaefer@tu-berlin.de}
	       \end{tabular}
	  \end{tabular}
}
\thanks{This work was supported by the German Research Foundation (DFG)
under Grants FR 4209/1-1 and SCHA 1944/7-1.}}

\begin{document}
\IEEEoverridecommandlockouts
\maketitle

\begin{abstract}
Deep learning based physical layer design, i.e., using dense neural networks as encoders and decoders, has received considerable interest recently. However, while such an approach is naturally training data-driven, actions of the wireless channel are mimicked using standard channel models, which only partially reflect the physical ground truth. Very recently, neural network based mutual information (MI) estimators have been proposed that directly extract channel actions from the input-output measurements and feed these outputs into the channel encoder. This is a promising direction as such a new design paradigm is fully adaptive and training data-based. This paper implements further recent improvements of such MI estimators, analyzes theoretically their suitability for the channel coding problem, and compares their performance. To this end, a new MI estimator using a \emph{``reverse  Jensen''} approach is proposed.
\end{abstract}

\section{Introduction}
Machine learning and in particular deep learning techniques, i.e., the use of neural networks, is an emerging tool for solving the communication task of noise-robust encoding and decoding of messages. Some of the latest advances involve an end-to-end view of the whole communication chain, where encoding and decoding are learnt simultaneously based on the concept of an \emph{autoencoder} \cite{OShea2017}. An autoencoder maps the input to the output, conditioned on an in-between constraint. For communication, this constraint becomes the communication channel itself. The most basic form of this constraint is given by additive noise on top of the signal that has been sent. Thus, if the channel model is known, one can simultaneously learn appropriate encoding and decoding so that the input message of the network gets encoded robustly, which then also enables correct decoding. The drawback is the required knowledge of the channel model, which needs to be known in advance to train the model. Moreover, the full channel model is required since the autoencoder learns by back-propagating through the channel noise layer, i.e., it needs to find the derivatives of the noise layer. 

There are several approaches to extend this framework to unknown channels including training the network on a generic channel model, such as the additive white Gaussian noise (AWGN) channel, and only tune the resulting decoder online \cite{SBrink2018}. Another approach is to circumvent the issue by using an action and rewards framework which works completely without knowledge of the underlying channel model and only takes into account the immediate rewards of the learned strategies by using \emph{reinforcement learning (RL)} \cite{aoudia2018end, goutay2018deep}. This has the advantage of its flexibility as it learns the basic principles for communication by itself using only training samples from the feedback link. 

However, such approaches are more sample-inefficient than other techniques, which can be augmented by expert knowledge to guide the process and learn from less information. For example, from a communication theoretic perspective, the underlying transition probability of the channel model is the key property which determines the communication rate. It is therefore reasonable to learn this probability distribution or some function of it, to guide the learning process of the communication algorithm. One approach in this direction is to estimate the underlying probability distribution of the channel from samples and model the channel layer based on this approximation. A particular successful approach to estimate and also generate distributions from samples are \emph{generative adversarial networks (GANs)} as introduced in \cite{goodfellow2014generative}. GANs are composed of two competing neural networks (NNs), i.e., a generative and a discriminative one. The generative NN tries to transform a noise input to look like the real data distribution, whereas the discriminative NN compares the samples of the real distribution (from the data) to the fake generated distribution and tries sort out the real from the fake samples. The generative NN therefore learns to imitate the real underlying distribution and can be used as a channel model layer \cite{ye2018GAN,oShea2018GAN}. Note that in this case, an end-to-end learning approach is still applicable in the end after successfully modeling the underlying channel. However, being able to estimate the probability distribution between input and output of the channel allows for a more powerful approach. One can decouple encoder and decoder, and learn the perfect encoder that maximizes the communication rate under the particular noise distribution constraint. In other words, one can use the conditional probability distribution to compute the mutual information (MI) and optimize the system such that it maximizes the MI between the channel input and the channel output. In fact, recent advances show that for additive noise channels one can also skip the in-between step of estimating and generating the channel distribution. It suffices to directly estimate the resulting MI from channel samples and use this to learn the MI maximizing encoding function \cite{FritschekMI19}. This does not only exploit the fact that the probability distribution of the channel is a key factor, but also that the MI is the actual target function to optimize for. 

Even though recent advances in estimating the MI from samples with variational methods combined with neural networks show remarkable success, state-of-the-art estimators have the drawbacks of having either a high bias or a high variance and finding a good estimator for controllable bias and variance remains an open problem. As many communication scenarios follow specific structured channel laws, the question is which estimator works best for these scenarios. In this paper, we therefore investigate the latest state-of-the-art estimators and compare their performances in consideration of specific generic communication channels. We outline a general framework for the estimator design and take this forward to a new variant based on a \emph{``reverse Jensen approach''}. The framework is illustrated with several examples and simulations.
\section{Communication System Model}

We consider a communication model with a transmitter, a channel, and a receiver. The transmitter wants to send a message $m\!\in\! \mathcal{M}\!=\!\{1,2, \ldots, 2^{nR}\}$ at a rate $R$ over a noisy channel using an encoding function $f(m)=x^n(m)\in\mathbb{C}^n$ to make the transmission robust against noise. Moreover, for every message $m\in\mathcal{M}$, we assume an average power constraint $\tfrac{1}{n}\sum_{i=1}^n|x_i(m)|^2\leq P$ on the corresponding codewords $x^n(m)$. The channel can be defined as the transition probability density $p_{Y^n|X^n}(y^n|x^n)$ for input and output sequences $x^n$ and $y^n$. If the channel is further memoryless, one has 
$
    p_{Y^n|X^n}(y^n|x^n)=\prod_{i=1}^{n} p(y_i|x_i),
$
 i.e., the output at time instant $i$ depends only on the corresponding input at time instant $i$ and is independent of the previous inputs.
The receiver uses a decoder $g(y^n)=\hat{m}$ to estimate and recover the original message. Moreover, the block error rate $P_e$ is defined as the average probability of error over all messages
$
    P_e= \frac{1}{|\mathcal{M}|}\sum_{m=1}^{|\mathcal{M}|} \mbox{Pr}(\hat{M}\neq m | M = m).
$
The general problem is now to find the maximal communication rate $R$, such that the error $P_e$ can be made arbitrary small for a sufficiently large $n$. This optimal rate $R$ is called the capacity of the channel and is known to be $C=\max_{p(x)} I(X;Y)$.

\section{Neural Estimation of Mutual Information}

Accurate MI estimation is a long standing problem which received renewed interest in the last years due to its application in the field of deep learning, for example in representation learning or the bottleneck hypothesis. The difficulty of estimating the MI stems from its dependence on the underlying joint probability density, which is unknown in most applications. Classical approaches to estimate the MI are based on binning the probability space \cite{fraser1986independent,darbellay1999estimation}, $k$-nearest neighbor statistics \cite{kraskov2004estimating,gao2015efficient,GaoEstimatorMI}, maximum likelihood estimation \cite{suzuki2008approximating}, and variational lower bounds \cite{barber2003algorithm}. Recently, there has been a surge of papers investigating the variational approach in combination with deep learning methods. These methods introduce a parametric function $T_\theta$, parametrized by the weights $\theta\in \Theta$ of a neural network, acting on the estimated densities. In particular, it can be viewed as a parametric estimate of a density ratio of the underlying distributions \cite{song2019understanding}.

\subsection{Overview of MI Estimators}
\label{estimators}

We will now briefly introduce and discuss some of the latest estimators. Let $\mathcal{F}$ be a family of functions $T_\theta: \mathcal{X} \times \mathcal{Y}\rightarrow \mathbb{R}$ parametrized by the weights $\theta \in \Theta$ of a neural network, then we have the following estimators:

\begin{definition}[MINE \cite{belghazi2018mine}]
\begin{equation}
    I_{\text{MINE}}= \sup_{f\in \mathcal{F}} \mathbb{E}_{p(x,y)}^{}[f(x,y)]-\log \mathbb{E}^{}_{p(x)p(y)}[e^{f(x,y)}]
\end{equation}
\end{definition}

This estimator utilizes the Donsker-Varadhan representation of the Kullback-Leibler divergence, which is a lower bound on the mutual information, and it is shown to converge to the true value of $I$ for increasing sample size $k$. However, note that due to the logarithm outside of the expectation in the second term, Monte Carlo sampling will introduce a bias, which yields neither a lower nor an upper bound on the true MI. In \cite{belghazi2018mine} it was proposed to use an exponential moving average on $\exp(f(X,Y))$ over mini-batches to reduce the bias.

\begin{definition}[NWJ \cite{nguyen2010estimating}]
\begin{equation}
    I_{\text{NWJ}}= \sup_{f\in \mathcal{F}} \mathbb{E}_{p(x,y)}^{}[f(x,y)]- \mathbb{E}^{}_{p(x)p(y)}[e^{f(x,y)-1}]
\end{equation}
\end{definition}
This is a lower bound and can be derived by using the Fenchel duality to bound the $f$-divergence from below. Using the conjugate dual function $f^*=\exp(x-1)$, one obtains a lower bound on the KL-divergence, which leads to this estimator. This is also known as the f-GAN objective \cite{nowozin2016f}. An alternative derivation of this estimator is shown in \cite{poole2018variational}. Note that, in contrast to MINE, the NWJ estimator is unbiased and yields a lower bound on the MI.

\begin{definition}[NCE \cite{oord2018representation}]
\begin{equation}
    I_{\text{NCE}}= \mathbb{E}_{p^K (x,y)}^{}\left[ \tfrac{1}{K}\sum_{i=1}^K \log \frac{e^f(x_i,y_i)}{\tfrac{1}{K} \sum_{j=1}^K e^{f(x_i,y_j)} }\right]
\end{equation}
\end{definition}
This estimator was introduced in the context of representation learning and was derived via noise-contrastive estimation (NCE). It differs from the first two estimators in the sense that it uses negative samples from the marginal distributions. It can be thought of as the categorical cross-entropy of the softmax of $f$, i.e., classifying a positive sample of $f$ correctly. The estimator exhibits low variance at the cost of high bias in the form of the upper bound $\log K$.

\begin{definition}[SMILE \cite{song2019understanding}]
\begin{equation}
    I_{\text{SMILE}}\!=\!\sup_{f\in \mathcal{F}} \mathbb{E}_{p(x,y)}^{}[f(X,Y)]-\log \mathbb{E}^{}_{p(x)p(y)}[c(e^{f(X,Y)},e^{-\tau})]
\end{equation}
\end{definition}
This estimator uses a clipping function $c=\text{clip}(u, v) = \max(\min(u, v), v) $ to constraint the expected value of the marginals in the second term. This reduces the variance in the estimator but introduces some biases. The estimator converges to $I_{\text{MINE}}$ for $\tau \rightarrow \infty$. 

Finally, we also introduce a new estimator in the following using a \emph{``reverse Jensen (RJ)''} approach for the partition function together with a box constraint on $\mathcal{F}$ (instead of the clipping function in SMILE). 
\begin{definition}
[RJE] 
\begin{align*}
I_{\text{RJE}} &  =\sup_{f\in\mathcal{F}_{\tau}}\mathbb{E}_{p(x,y)}^{{}%
}[f(x,y)]\\
& \qquad -\min_{a>b}\frac{a\mathbb{E}_{p(x)p(y)}^{{}}\left[  \log(1+ae^{f(x,y)}%
)\right]  }{(1-(b/a)^{1/2})^{+}}+\log\left(  a\right)
\end{align*}
\end{definition}

Here, $b\geq1$ depends on the critic and the marginals, and $\mathcal{F}%
_{\tau}$ is the set of critics bounded by $\tau$. The estimator is somewhat difficult to tune due to the multiple parameters but, ideally, it is a compromise between the resulting bias and variance, similar to SMILE, and also provides a strict lower bound on the MI.

The frameworks and specific properties of the estimators are discussed in detail in the
next section. For simple exposition, let us identify $p(x,y)$ with the ground truth
probability measure $\mathbb{P}$ and $\mathbb{Q}:=\mathbb{P}_{x}\times\mathbb{P}%
_{y}$ with the marginals $p(x)$ and $p(y)$, respectively. Further, let
$\mathbb{P}^{n}$ and $\mathbb{Q}^{n}$ denote the empirical measures from a set of i.i.d. samples. Let
$\mathbb{G}$ be any positive measure with total variation $|\mathbb{G}|=\int d\mathbb{G}$.
\subsection{Discussion and Analysis}
\label{discussion}

\subsubsection{MI from unnormalized Gibbs measures}

To start with, MINE seeks to estimate a critic $f$ with support on the
ground truth measures $\mathbb{P}_x$ and $\mathbb{P}_y$ that dominate $\mathbb{P}$. To verify the optimality
of the estimator, we can identify with any $f$ the so-called Gibb's measure
$d\mathbb{G}=\frac{e^{f(x,y)}}{\mathbb{E}_{\mathbb{Q}}[e^{f(x,y)}]}d\mathbb{Q}$.
Obviously, by construction, $\mathbb{Q}$ dominates $\mathbb{G}$ and
$\mathbb{G}$ is a probability measure. Hence, we have%
\begin{equation}
\mathbb{E}_{\mathbb{P}}\log\frac{d\mathbb{G}}{d\mathbb{Q}} =\mathbb{E}%
_{\mathbb{P}}[f(x,y)]-\log\mathbb{E}_{\mathbb{Q}}[e^{f(x,y)}] 
\leq\mathbb{E}_{\mathbb{P}}\log\frac{d\mathbb{P}}{d\mathbb{Q}}%
\label{eq:ineq}
\end{equation}
where the inequality is due to $|\mathbb{P}|=|\mathbb{G}|=1$ and the positiveness of 
the Kullback-Leibler divergence $D_{\text{KL}}(\mathbb{P}\|\mathbb{G})$. Moreover, a non-unique optimum is $f^{\ast}=\log
\frac{d\mathbb{P}}{d\mathbb{Q}}+c$, i.e., the estimate of MINE can be
unnormalized. However, they suffer from an unbiased estimate, since
$\mathbb{E}_{\mathbb{Q}}[\log\mathbb{E}_{\mathbb{Q}^{n}}[e^{f(x,y)}]]\neq
\log\mathbb{E}_{\mathbb{Q}}[e^{f(x,y)}]$ which can be observed in the simulations.
A lower bound of the variance $\mathbb{V}_{\mathbb{G},\mathbb{Q}}$ of the partition function estimator $\mathbb{E}_{\mathbb{Q}}[e^{f(x,y)}]$ can be given as
\begin{equation*}
\liminf_{n}n\mathbb{V}_{\mathbb{G},\mathbb{Q}}[\log\mathbb{E}_{\mathbb{Q}^{n}}%
[e^{f(x,y)}]]\geq e^{D_{\text{KL}}(\mathbb{G}\|\mathbb{Q})}-1    
\end{equation*}
which is a straightforward extension of \cite[Theorem 2]{song2019understanding}. Due to the independence of ${\mathbb{Q}^{n}}$ and ${\mathbb{P}^{n}}$ it easily follows that for any (suboptimal) estimate $\mathbb{G}$ (through critic $f$) and in the optimum $\mathbb{P}=\mathbb{G}$ (i.e. optimal $f^*$) the variance of MINE (and also NWJ, see below) scales
exponentially with the estimated mutual information which can be clearly observed in all simulation examples that we have done.

An interesting new direction is obtained when we identify $f$ with the family of unnormalized Gibbs measures.

\begin{definition}
An unnormalized Gibbs measure is defined by%
\[
d\mathbb{G}=\frac{e^{f(x,y)}}{G\left(  f\right)  }d\mathbb{Q},\;\exists\,
c\in\mathbb{R}:\int\frac{e^{f(x,y)+c}}{G\left(  f+c\right)  }d\mathbb{Q}=1
\]
where $G$ is some normalization function.
\end{definition}

One example is actually the NWJ estimator where%
\[
d\mathbb{G}=\frac{e^{f(x,y)}}{\exp\left(  e^{-1}\mathbb{E}_{\mathbb{Q}%
}[e^{f(x,y)}]\right)  }d\mathbb{Q}.%
\]
In general we have here $\mathbb{G}>1$ or $\mathbb{G}<1$. The trick is to show the inequality (\ref{eq:ineq}) in
a different way. For the NWJ estimator, we have%
\begin{align*}
\mathbb{E}_{\mathbb{P}}\log\frac{d\mathbb{G}}{d\mathbb{Q}} &  =\mathbb{E}%
_{\mathbb{P}}[  f\left(  x,y\right)  ]  -\frac{1}{e}\mathbb{E}%
_{\mathbb{Q}}[e^{f(x,y)}]\\
&  \leq\mathbb{E}_{\mathbb{P}}[f(x,y)]-\log\mathbb{E}_{\mathbb{Q}}%
[e^{f(x,y)}]  \leq\mathbb{E}_{\mathbb{P}}\log\frac{d\mathbb{P}}{d\mathbb{Q}}%
\end{align*}
by the simple inequality $\log\left(  x\right)  \leq\frac{x}{e}$. Notably, the
measures identified with $f$ in the first and second line are actually
different, but all inequalities become tight for $f^{\ast}=\log\frac
{d\mathbb{P}}{d\mathbb{Q}}+1$. In our simulations We have also found that such ``self-normalization'' property \cite{poole2018variational} seems to cause no problems in our coding scenario. The
estimator is now in fact unbiased but due to the simple bounding technique the variance of this estimator depends linearly on the partition function estimator
$\mathbb{E}_{\mathbb{Q}}[e^{f(x,y)}]$. We can show that
\[
\mathbb{V}_{\mathbb{G},\mathbb{Q}}[\mathbb{E}_{\mathbb{Q}^{n}}[e^{f(x,y)}]]\geq
\frac{e^{D_{KL}(\mathbb{G}||\mathbb{Q})}-|\mathbb{G}|^{2}}{e^{|\mathbb{G}|}n}%
\]
which, again, is a straightforward extension of \cite[Theorem~2]{song2019understanding}.
The result suggests, somewhat surprising, a slightly smaller variance of NWJ which is verified in the AWGN simulations.

In the following section we ask whether or not the two properties, i.e., unbiased
estimate and lower variance of the estimation, can be combined in some way. Notably,
\cite{song2019understanding} has addressed this issue and proposed SMILE
which simply bounds the variance of the partition function estimator as follows%
\[
\mathbb{V}_{\mathbb{P},\mathbb{Q}}[\mathbb{E}_{\mathbb{Q}^{n}}[e^{f(x,y)}]]\leq
\frac{e^{\tau}-e^{-\tau}}{4n}.%
\]
On the other hand, since $f$ is clipped, the new identified measure is
$d\mathbb{G}=\frac{e^{f(x,y)}}{\exp\left(  \mathbb{E}_{\mathbb{Q}}[c(e^{f(x,y)}%
,e^{-\tau},e^{\tau})]\right)  }d\mathbb{Q}$ so that $\mathbb{E}_{\mathbb{P}}%
\log\frac{d\mathbb{G}}{d\mathbb{Q}}$ and $\mathbb{E}_{\mathbb{P}}\log
\frac{d\mathbb{P}}{d\mathbb{Q}}$ are essentially indifferent which means that SMILE can either over- or undershoot the true ground truth MI. A bound on the bias is also provided in \cite{song2019understanding} as%
\begin{align*}
&\left\vert \mathbb{E}_{\mathbb{Q}}[e^{f(x,y)}]-\mathbb{E}%
_{\mathbb{Q}} [c(e^{f(x,y)}%
,e^{-\tau},e^{\tau})]\right\vert  \\ 
&\qquad\qquad\leq \max \left(  e^{\tau
} -|\mathbb{G}|e^{-2\tau},|\mathbb{G}|-e^{-\tau}\right).
\end{align*}
A proper way of how to select the clip value is an open problem and a real
practical challenge. Therefore, in the following we use a different so-called reverse Jensen's inequality approach.

\subsubsection{Reverse Jensen Approach}
The new RJE approach is based on the following partial converse of Jensen's inequality.
\begin{lemma}
\label{lem:jensen}
For any random variable $X \geq 0$, it holds%
\[
\log{(}\mathbb{E}\left[  {X}\right]  {)}\leq\min_{a>b}\left(  \frac
{a\mathbb{E}\left[  \log(1+a{X})\right]  }{\left(  1-\sqrt{\frac{b}{a}%
}\right)  }-\log(a)\right)
\]
where $b:=\mathbb{E}\left[  {X^{2}}\right]  /\mathbb{E}\left[
{X}\right]  ^{2}<\infty$.
\end{lemma}

The proof of the lemma is omitted due to lack of space.

Define the non-centralized moment with respect to $\mathbb{Q}$ as
$m_{i}(f):=\mathbb{E}_{\mathbb{Q}}[e^{if(x,y)}]$. We have the following theorem.

\begin{theorem}
We have $
I(X;Y)\geq I_{\text{RJE}}$
provided $f^{\ast}\in\mathcal{F}_{\tau}$ is such that $b\geq\frac{m_{2}(f^{\ast}%
)}{m_{1}^{2}(f^{\ast})}$. Moreover, the second moment is lower bounded as %
$m_{2}(f^{\ast})\geq e^{D_{\text{KL}}(\mathbb{G}^{\ast}\|\mathbb{Q})}\geq
e^{\mathbb{E}_{\mathbb{P}}\log\frac{d\mathbb{G}^{\ast}}{d\mathbb{Q}}}$.%

\end{theorem}

\begin{proof}
The bound on the MI is a direct consequence of a chain of inequalities similarly as in (\ref{eq:ineq}) and Lemma \ref{lem:jensen}. The lower bound on the second moment can be proved by a change of measure in the Radon-Nikodym derivative and the positiveness of the KL divergence.  \end{proof}
\vspace*{0.5\baselineskip}

Notably, a bound on the variance can be obtained in a straightforward way by using the ``delta method'' as in MINE but which now depends on $a$ and $b$ (omitted to space limitations). Due to the improved bounding technique from the reverse Jensen's inequality, we expect a smaller variance compared to MINE (and NWJ), which is indeed verified in the simulations.
It must be noted that a critical issue left for future investigation is the bias which also depends on the parameter setting and the actual applied algorithm. Finally, we mention that NCE falls in the general framework but is not competitive for high MI values due to its upper bound $\log n$.

\section{Implementation}
We have implemented the MI estimators with a neural network with two hidden layers, each comprised of 256 nodes and ReLU activation functions. For the input we use a joint critic, rather than a separable critic, since this was shown to yield results with less variance \cite{poole2018variational}. This means that both input samples from $X$ and $Y$ get concatenated and fed into the network. Moreover, we use all marginal samples instead of a shifted version only as in \cite{poole2018variational}. To this end, we draw $K$ samples from the joint distribution $p(x,y)$ and then use all pairs $(x_i,y_j)$, $i\neq j$, for the marginals. This yields $K(K-1)$ samples from the marginals instead of only $K$, which we would obtain from sampling a shifted marginal, i.e., $(x_i, y_{i+1})$. All expectations are replaced by the sample average over a mini-batch. The batch size is chosen to be $64$ for all calculations. Due to space limitations, we only show results for the two simplest channel models: the AWGN channel and the binary symmetric channel (BSC) for continuous and discrete inputs.

\begin{example}
[AWGN]
Let $X,Z \in \mathbb{R}^d$ be independent Gaussian random variables with $X \sim \mathcal{N}(0,I\sigma_x^2)$ and $Z \sim \mathcal{N}(0,I\sigma_z^2)$, and $Y=X+Z$, then the mutual information is given as
$
    I(X;Y)=\tfrac{d}{2} \log (1+\tfrac{\sigma_x^2}{\sigma_z^2}).
$
\end{example}

\begin{example}[BSC]
Let $X\sim$ Bern$(\tfrac{1}{2})$ and $Z \sim$ Bern($\delta$), then $Y=X+Z$ and the mutual information is given by
$
    I(X;Y)=H(Y)-H(Z)=1-h_b(\delta),
$ where $h_b$ denotes the binary entropy function.

\end{example}
We have compared the estimators for the AWGN channel in Fig.~\ref{fig:MI_Simu_Gauss} and for the BSC in Fig.~\ref{fig:MI_Simu_BSC}. (other channel models including Rayleigh are omitted due to space constraints).

To test the learning ability of the channel encoding, we generate $16$ messages uniformly and send them through the initialized encoder, which generates $X^n$ according to $p(x^n,y^n)$. The corresponding samples of $Y^n$ are generated by our AWGN channel, where the noise variance $\sigma_z^2$ is scaled such that we have a resulting signal-to-noise ratio per bit of $7$ $E_b/N_0$ [db]. Note also that the encoded signal $X^n$ has a unit average power normalization $\mathbb{E}(|X_i|^2)=1$, where the expectation is over the signal dimension and the batch size.
The training procedure is similar to \cite{FritschekMI19}, where we alternate between maximizing weights of the estimator $\theta$ and the encoder weights $\phi$ over $\max_{\phi}\max_{\theta} \tilde{I}_{\theta}(X^n_{\phi}(m);Y^n)$. The MI estimator is initially trained with $500$ iterations and batch size $64$. Afterwards we train the encoder for $5$ epochs with $400$ iterations and batch size $64$. After each epoch, we tune the MI estimator with one iteration with batch size $64$. In the end, the decoder is trained for $5$ epochs, with $400$ iterations. During the whole procedure, the learning rate is kept fixed at $0.005$ with the NADAM optimizer. We note that we have not put particular emphasis on optimizing the parameters, for which we expect further improvements. The results are shown in Fig.~\ref{fig:Enc_Gauss} (note that results for the Rayleigh channel are similar, i.e., no gap between estimator performance).
The simulation code is available at \cite{Fritschek2020_code}, implemented with TensorFlow 2.1 \cite{tensorflow2015-whitepaper}. 


\begin{figure}
    \centering
    \includegraphics[scale=0.4]{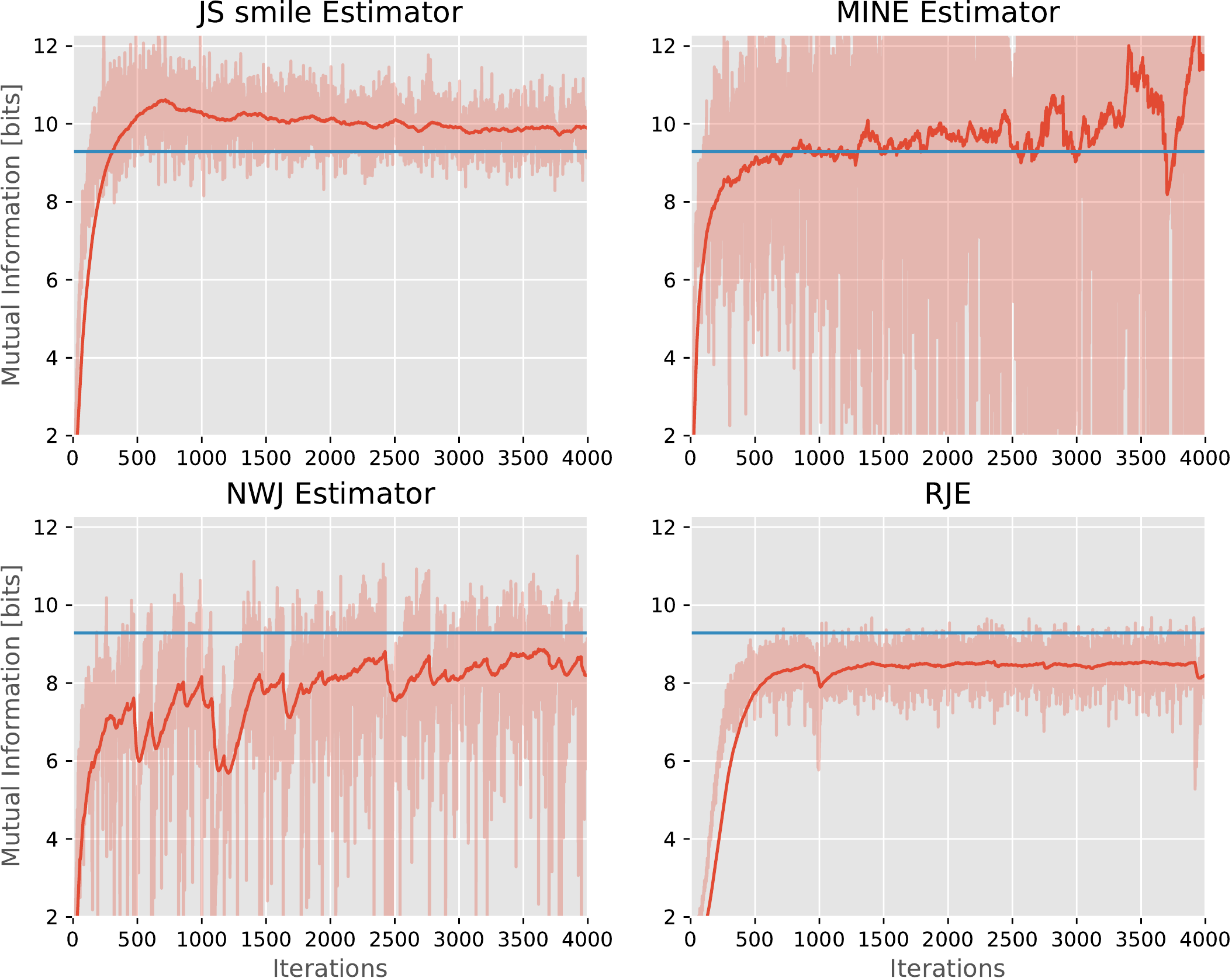}
    \caption{MI estimators for the $8$-dim. AWGN channel with SNR$=4$. Note that NCE runs into its upper bound $\log 64 = 6$ and is omitted. RJE is implemented with $\tau=6$ and $a=2b$.}    \label{fig:MI_Simu_Gauss}
\end{figure}

\begin{figure}
    \centering
    \includegraphics[scale=0.4]{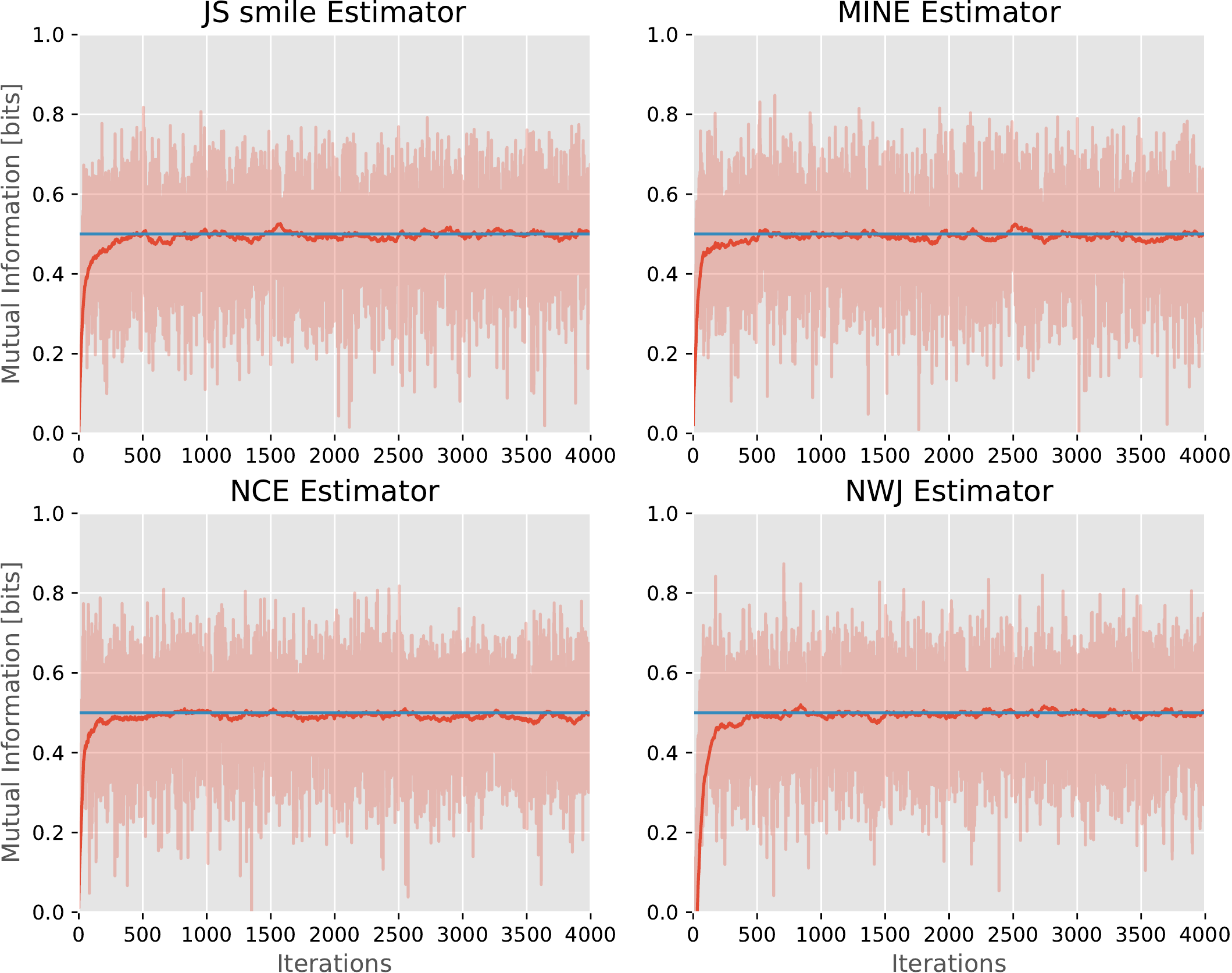}
    \caption{MI estimators for the BSC with $\delta=0.11$, which results in $I \approx 0.5$. SMILE is implemented with $\tau=5$.}
    \label{fig:MI_Simu_BSC}
\end{figure}

\begin{figure}
    \centering
    \includegraphics[scale=0.45]{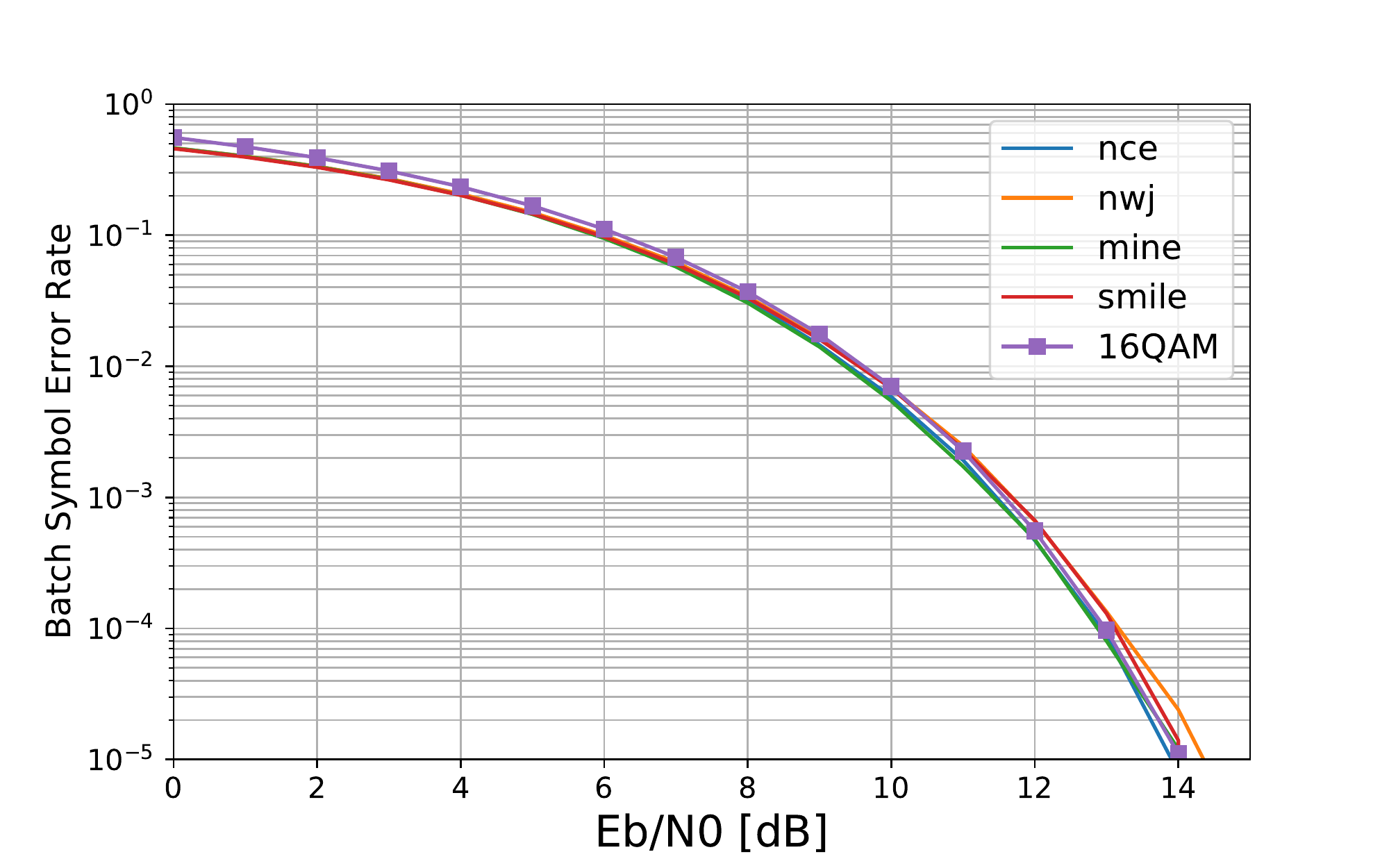}
    \caption{Performance of the learned encoder and decoder pair for an AWGN channel with $n=2$ and $16$ messages input.}
    \label{fig:Enc_Gauss}
\end{figure}

\section{Conclusions and outlook}
In this paper we have investigated variational MI estimation approaches for channel coding. We have seen that the estimators show quite different behaviors in terms of bias and variance for classical channel models. The proposed RJE provides an excellent tradeoff in this regard. Surprisingly, these different behaviors do not affect the performance in the channel coding problem, where all estimators perform quite robustly. One reason might be that we have limited the encoding simulation to $16$ messages with $2$ dimensional variables, wheres the MI estimation simulation is run on Gaussian inputs for $8$-dimensional variables. Consequently, the fastest procedure can be taken which we believe is therefore an excellent alternative to competitive approaches such as reinforcement learning \cite{goutay2018deep}.

\medskip

\small

\bibliographystyle{./IEEEtran}
\bibliography{./ref}

\begin{thebibliography}{10}
\providecommand{\url}[1]{#1}
\csname url@samestyle\endcsname
\providecommand{\newblock}{\relax}
\providecommand{\bibinfo}[2]{#2}
\providecommand{\BIBentrySTDinterwordspacing}{\spaceskip=0pt\relax}
\providecommand{\BIBentryALTinterwordstretchfactor}{4}
\providecommand{\BIBentryALTinterwordspacing}{\spaceskip=\fontdimen2\font plus
\BIBentryALTinterwordstretchfactor\fontdimen3\font minus
  \fontdimen4\font\relax}
\providecommand{\BIBforeignlanguage}[2]{{%
\expandafter\ifx\csname l@#1\endcsname\relax
\typeout{** WARNING: IEEEtran.bst: No hyphenation pattern has been}%
\typeout{** loaded for the language `#1'. Using the pattern for}%
\typeout{** the default language instead.}%
\else
\language=\csname l@#1\endcsname
\fi
#2}}
\providecommand{\BIBdecl}{\relax}
\BIBdecl

\bibitem{OShea2017}
T.~O'Shea and J.~Hoydis, ``An introduction to deep learning for the physical
  layer,'' \emph{IEEE Trans. on Cogn. Commun. Netw.}, vol.~3, no.~4, pp.
  563--575, Dec. 2017.

\bibitem{SBrink2018}
S.~D{\"o}rner, S.~Cammerer, J.~Hoydis, and S.~ten Brink, ``Deep learning based
  communication over the air,'' \emph{IEEE J. Sel. Topics Signal Process.},
  vol.~12, no.~1, pp. 132--143, Feb. 2018.

\bibitem{aoudia2018end}
F.~A. Aoudia and J.~Hoydis, ``End-to-end learning of communications systems
  without a channel model,'' in \emph{Proc. 52nd Asilomar Conf. Signals,
  Systems, Computers}, Pacific Grove, CA, USA, Oct. 2018, pp. 298--303.

\bibitem{goutay2018deep}
M.~Goutay, F.~A. Aoudia, and J.~Hoydis, ``Deep reinforcement learning
  autoencoder with noisy feedback,'' \emph{Preprint arXiv:1810.05419}, 2018.

\bibitem{goodfellow2014generative}
I.~Goodfellow, J.~Pouget-Abadie, M.~Mirza, B.~Xu, D.~Warde-Farley, S.~Ozair,
  A.~Courville, and Y.~Bengio, ``Generative adversarial nets,'' in
  \emph{Advances in neural information processing systems}, 2014, pp.
  2672--2680.

\bibitem{ye2018GAN}
H.~Ye, G.~Y. Li, B.-H.~F. Juang, and K.~Sivanesan, ``Channel agnostic
  end-to-end learning based communication systems with conditional {GAN},'' in
  \emph{Proc. IEEE Global Commun. Conf. Workshops}, Abu Dhabi, United Arab
  Emirates, Dec. 2018, pp. 1--5.

\bibitem{oShea2018GAN}
T.~J. O'Shea, T.~Roy, N.~West, and B.~C. Hilburn, ``Physical layer
  communications system design over-the-air using adversarial networks,'' in
  \emph{Proc. 26th European Signal Process. Conf.}, Rome, Italy, Sep. 2018, pp.
  529--532.

\bibitem{FritschekMI19}
R.~{Fritschek}, R.~F. {Schaefer}, and G.~{Wunder}, ``Deep learning for channel
  coding via neural mutual information estimation,'' in \emph{Proc. 20th IEEE
  Int. Workshop Signal Process. Adv. Wireless Commun.}, Cannes, France, Jul.
  2019, pp. 1--5.

\bibitem{fraser1986independent}
A.~M. Fraser and H.~L. Swinney, ``Independent coordinates for strange
  attractors from mutual information,'' \emph{Physical Review A}, vol.~33,
  no.~2, p. 1134, 1986.

\bibitem{darbellay1999estimation}
G.~A. Darbellay and I.~Vajda, ``Estimation of the information by an adaptive
  partitioning of the observation space,'' \emph{IEEE Trans. Inf. Theory},
  vol.~45, no.~4, pp. 1315--1321, May 1999.

\bibitem{kraskov2004estimating}
A.~Kraskov, H.~St{\"o}gbauer, and P.~Grassberger, ``Estimating mutual
  information,'' \emph{Physical Review E}, vol.~69, no.~6, p. 066138, 2004.

\bibitem{gao2015efficient}
S.~Gao, G.~Ver~Steeg, and A.~Galstyan, ``Efficient estimation of mutual
  information for strongly dependent variables,'' in \emph{Artificial
  Intelligence and Statistics}, 2015, pp. 277--286.

\bibitem{GaoEstimatorMI}
W.~Gao, S.~Oh, and P.~Viswanath, ``Demystifying fixed $k$-nearest neighbor
  information estimators,'' \emph{IEEE Trans. Inf. Theory}, vol.~64, no.~8, pp.
  5629--5661, Aug. 2018.

\bibitem{suzuki2008approximating}
T.~Suzuki, M.~Sugiyama, J.~Sese, and T.~Kanamori, ``Approximating mutual
  information by maximum likelihood density ratio estimation,'' in \emph{New
  challenges for feature selection in data mining and knowledge discovery},
  2008, pp. 5--20.

\bibitem{barber2003algorithm}
D.~Barber and F.~Agakov, ``The {IM} algorithm: A variational approach to
  information maximization,'' in \emph{Proc. 16th Int. Conf. Neural Information
  Processing Systems}.\hskip 1em plus 0.5em minus 0.4em\relax MIT Press, 2003,
  pp. 201--208.

\bibitem{song2019understanding}
J.~Song and S.~Ermon, ``Understanding the limitations of variational mutual
  information estimators,'' in \emph{Proc. 8th Int. Conf. Learning
  Representations}, Addis Ababa, Ethiopia, Apr. 2020.

\bibitem{belghazi2018mine}
I.~Belghazi, S.~Rajeswar, A.~Baratin, R.~D. Hjelm, and A.~Courville, ``{MINE}:
  Mutual information neural estimation,'' in \emph{Proc. 35th Int. Conf.
  Machine Learning}, Stockhom, Sweden, Jul. 2018.

\bibitem{nguyen2010estimating}
X.~Nguyen, M.~J. Wainwright, and M.~I. Jordan, ``Estimating divergence
  functionals and the likelihood ratio by convex risk minimization,''
  \emph{IEEE Trans. Inf. Theory}, vol.~56, no.~11, pp. 5847--5861, Nov. 2010.

\bibitem{nowozin2016f}
S.~Nowozin, B.~Cseke, and R.~Tomioka, ``$f$-{GAN}: Training generative neural
  samplers using variational divergence minimization,'' in \emph{Advances in
  Neural Information Processing Systems}, 2016, pp. 271--279.

\bibitem{poole2018variational}
B.~Poole, S.~Ozair, A.~van~den Oord, A.~A. Alemi, and G.~Tucker, ``On
  variational lower bounds of mutual information,'' in \emph{NeurIPS Workshop
  on Bayesian Deep Learning}, 2018.

\bibitem{oord2018representation}
A.~van~den Oord, Y.~Li, and O.~Vinyals, ``Representation learning with
  contrastive predictive coding,'' \emph{Preprint arXiv:1807.03748}, 2018.

\bibitem{Fritschek2020_code}
R.~Fritschek, ``Simulations {SPAWC} 2020,'' \url{https://github.com/Fritschek},
  2020.

\bibitem{tensorflow2015-whitepaper}
\BIBentryALTinterwordspacing
J.~Dean, R.~Monga \emph{et~al.}, ``{TensorFlow}: Large-scale machine learning
  on heterogeneous systems,'' 2015, software available from tensorflow.org.
  [Online]. Available: \url{https://www.tensorflow.org/}
\BIBentrySTDinterwordspacing

\end{thebibliography}

\end{document}